\documentclass[conference]{IEEEtran}
\IEEEoverridecommandlockouts
\usepackage{cite}
\usepackage{amsmath,amssymb,amsfonts}
\usepackage{algorithmic}
\usepackage{float}
\usepackage{graphicx}
\usepackage{textcomp}
\usepackage{xcolor}
\usepackage{subfigure}
\usepackage{array}
\usepackage{multicol}
\usepackage{multirow}
\allowdisplaybreaks

\newtheorem{corollary}{Corollary}

\newtheorem{lemma}{Lemma}
\newtheorem{theorem}{Theorem}
\newenvironment{proof}{{\ \ \noindent\it Proof:\ }}

\def\BibTeX{{\rm B\kern-.05em{\sc i\kern-.025em b}\kern-.08em
    T\kern-.1667em\lower.7ex\hbox{E}\kern-.125emX}}
\begin{document}

\title{Learning Based Joint Coding-Modulation for Digital Semantic Communication Systems
\thanks{This work is supported by the NSF of China under grant 62125108, 61901261, 12031011.}
}

\makeatletter
\newcommand{\linebreakand}{%
  \end{@IEEEauthorhalign}
  \hfill\mbox{}\par
  \mbox{}\hfill\begin{@IEEEauthorhalign}
}
\makeatother

\author{\IEEEauthorblockN{Yufei Bo}
\IEEEauthorblockA{\textit{Department of Electronic Engineering} \\
\textit{Shanghai Jiao Tong University}\\
Shanghai, China \\
boyufei01@sjtu.edu.cn}
\and
\IEEEauthorblockN{Yiheng Duan}
\IEEEauthorblockA{\textit{Department of Electronic Engineering} \\
\textit{Shanghai Jiao Tong University}\\
Shanghai, China \\
duanyiheng@sjtu.edu.cn}
\linebreakand 
\IEEEauthorblockN{Shuo Shao}
\IEEEauthorblockA{\textit{School of Cyber Science and Engineering} \\
\textit{Shanghai Jiao Tong University}\\
Shanghai, China \\
shuoshao@sjtu.edu.cn}
\and
\IEEEauthorblockN{Meixia Tao}
\IEEEauthorblockA{\textit{Department of Electronic Engineering} \\
\textit{Shanghai Jiao Tong University}\\
Shanghai, China \\
mxtao@sjtu.edu.cn}
}

\maketitle

\begin{abstract}
In learning-based semantic communications, neural networks have replaced different building blocks in traditional communication systems. However, the digital modulation still remains a challenge for neural networks. The intrinsic mechanism of neural network based digital modulation is mapping continuous output of the neural network encoder into discrete constellation symbols, which is a non-differentiable function that cannot be trained with existing gradient descend algorithms. To overcome this challenge, in this paper we develop a joint coding-modulation scheme for digital semantic communications with BPSK modulation. In our method, the neural network outputs the likelihood of each constellation point, instead of having a concrete mapping. A random code rather than a deterministic code is hence used, which preserves more information for the symbols with a close likelihood on each constellation point. The joint coding-modulation design can match the modulation process with channel states, and hence improve the performance of digital semantic communications. Experiment results show that our method outperforms existing digital modulation methods in semantic communications over a wide range of SNR, and outperforms neural network based analog modulation method in low SNR regime.   
\end{abstract}


\section{Introduction}
Recently, semantic communications are envisioned to be a key enabling technique for future 6G communication networks\cite{ping}. As a paradigm shift beyond Shannon, semantic communications only transmit necessary information relevant to the  specific task at the receiver, known as ``semantic information'', and therefore can significantly improve transmission efficiency. Semantic communications are regarded as a truly intelligent communication framework. They can find a wide range of applications such as smart transportation, video conference, augmented reality (AR), interactive hologram, and intelligent humanoid robots.

Thanks to the rapid development of machine learning, neural networks (NNs) are powerful tools to learn and extract the hidden semantic information and have been successfully applied in the design of various semantic communication systems. 
Specifically, neural networks are most widely used to replace the classic coding and decoding blocks.
For example, the neural network is used as source coding in the semantic communication system developed by Dai \emph{et al}. in \cite{coded_trans}. Neural networks are also used as joint source channel coding (JSCC) in \cite{8bit, gunduz1}. Meanwhile, neural networks can also achieve other functions such as repeated requesting. For instance, Jiang \emph{et al}.\cite{harq} integrate neural network based Hybrid Automatic Repeat Request (HARQ) mechanism into semantic communication systems.

\begin{figure*}[htb]
\centering
\setlength{\abovecaptionskip}{-0.2cm}
	\includegraphics[scale=0.14]{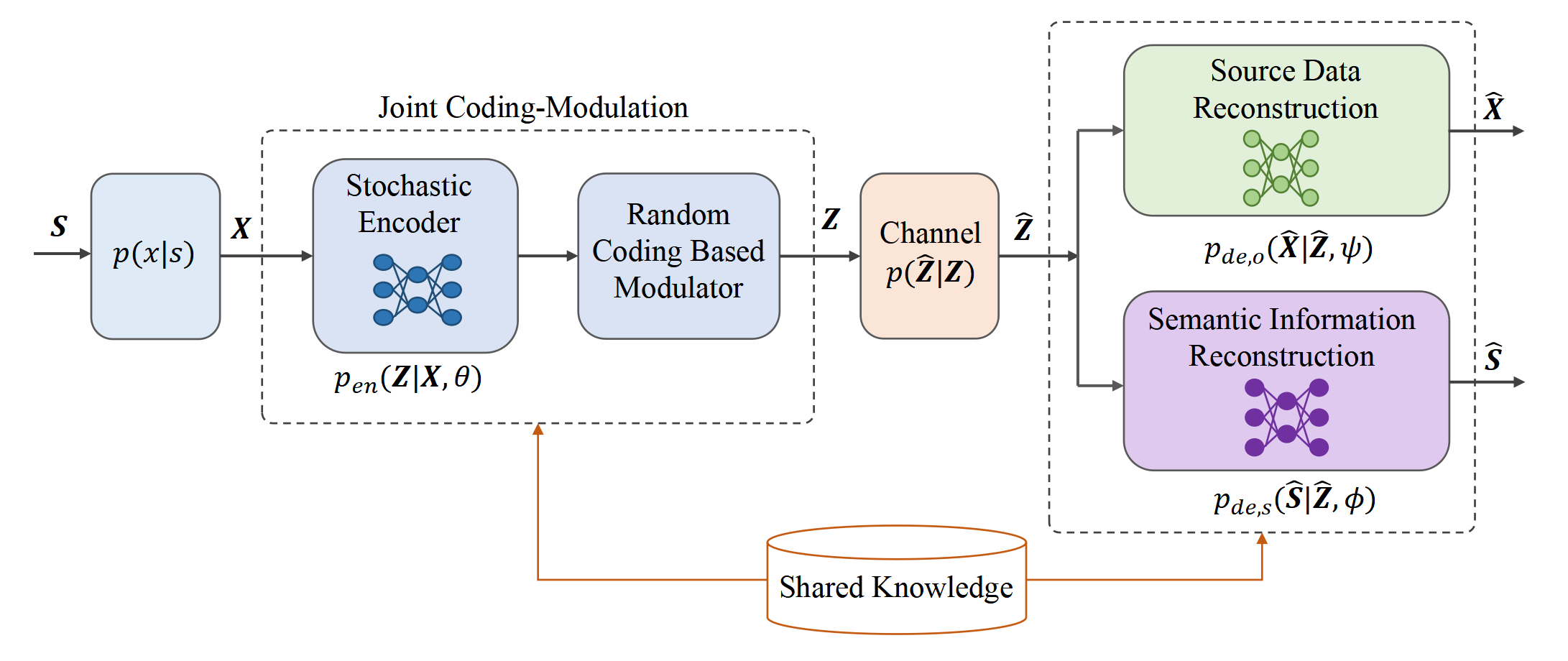}
\caption{The semantic communication system model.}
\vspace{-0.6cm}
\label{semantic}
\end{figure*}

Among different building blocks in modern wireless communications, digital modulation remains a difficult component to be modeled by neural networks. Analog modulation is used in most previous works on NN-based semantic communication systems\cite{gunduz1, speech1, zhang, vqa}, which does not match up with the real-world digital systems. The difficulty of realizing NN-based digital modulation comes from its intrinsic mechanism. The output of NN-based encoders is often real numbers, and the modulator needs to establish a mapping from these real numbers to discrete constellation symbols, which is equivalent to a non-differentiable function. Hence this modulator cannot be achieved by neural networks directly, since all the existing gradient descend algorithms will fail with the non-differentiable mapping function.

To the best of our knowledge, there are two different methods to achieve digital modulation in semantic communications. One straightforward idea is to use a non-neural-network-based method for digital modulation. For example, an 8-bit quantizer is used in \cite{8bit} to convert the continuous output of the JSCC encoder to discrete symbol sequences, which can be transmitted through digital communication systems. Another approach is to use an extra neural network to generate the likelihood of constellations instead of a concrete one. For example, Jiang \emph{et al}. \cite{harq} use an MLP network with Sigmoid function to train a non-uniform quantizer that has the minimum quantization error, and further associate the quantized bits sequence with the constellation $\left \lbrace+1, -1\right \rbrace$ for BPSK modulation.
However, in both quantization-based modulation methods, the modulator block works separately with the coding block, which will decrease the transmission performance especially when the channel noise situation is poor. 

In this paper, we propose an NN-based joint coding-modulation scheme for semantic communication systems with BPSK modulation. Specifically, the joint coding-modulation (JCM) is modeled by an NN-based stochastic encoder in conjunction with a random coding-based modulator. The NN learns the transfer probability from source data, while the random coding generates the actual modulated symbols for channel transmission based on the learned transfer probability. Without intermediate steps and jointly trained in one unique neural network, our method can better match up the coding and modulation process with channel states.

In addition, the capacity-approaching property of random coding can preserve more information for those symbols who has a marginal probability around $0.5$. Moreover, 
the Gumbel-max reparameterization method\cite{gumbel} is used in random code generation to alleviate the computation complexity of training the JCM network. Accordingly, at the receiver side, two separate NNs are used to reconstruct the source data and the semantic information, respectively, and they are trained jointly with the neural network at the transmitter end.

Based on the proposed JCM scheme, we develop an end-to-end semantic communication system using the Variational Autoencoder (VAE) architecture\cite{vae}. To this end, two theoretic contributions have been made on the loss function design for the VAE network.
First, we derive a variational inference lower bound for the mutual information based neural network optimization objective function. This lower bound is much easier to estimate and optimize than the mutual information function when data are high dimensional variables, and can be used in VAE networks. Second, we further extend the established lower bound to the scenario of image semantic communications, and derive a loss function that consists of cross entropy and mean square error distortion function.
We would like to point out that the meaning of theses two theoretic results are beyond the engineering target in this paper. Both the mutual information based\cite{nesct} and data distortion based\cite{gunduz1} loss functions are widely used in semantic communications, but existing semantic rate distortion model \cite{rate} only tells us that they are strongly related without telling us how exactly they are related. Our results reveal a more concrete insight that this relationship should be that one can lower bound the other in image semantic communications.   

Experiment results on real-world datasets show that our method can inherit the performance advantage of VAE in neural network based JSCC problems\cite{zhang}. Specifically, our method outperforms existing methods for digital semantic communications, namely the uniform quantizer \cite{8bit} and the neural network based quantizer \cite{harq}. Moreover, experiments also show that our method can have a close performance with neural network based semantic communication systems with analog modulation when the Signal-to-Noise-Ratio (SNR) is high, and a better performance when SNR is low.

\begin{figure*}[!t]
\normalsize
\setcounter{equation}{1}
\begin{align}
&I(\mathbf{S};\mathbf{\hat Z}) + \lambda\cdot I(\mathbf{\mathbf{X};\hat Z}) \ge
-\mathbb{E}_{p(\mathbf{\hat z}|\mathbf{z})}\mathbb{E}_{p_{en}(\mathbf{z}|\mathbf{x}, \theta)}\mathbb{E}_{p(\mathbf{s},\mathbf{x})} \lbrace \mathrm{CE}\left \lbrack p(\mathbf{s}|\mathbf{\hat{z}})||p_{de,s}(\mathbf{\hat s}|\mathbf{\hat z}, \phi)\right \rbrack  +\lambda\cdot \mathrm{CE}\left \lbrack p(\mathbf{x}|\mathbf{\hat{z}})||p_{de,o}(\mathbf{\hat x}|\mathbf{\hat z}, \psi)\right \rbrack\rbrace + K
    \label{elbo}\\
&I(\mathbf{S};\mathbf{\hat Z}) + \lambda\cdot I(\mathbf{\mathbf{X};\hat Z}) \ge
-\mathbb{E}_{p(\mathbf{\hat z})}
    \lbrace \mathrm{CE}\left \lbrack p(\mathbf{s}|\mathbf{\hat{z}})||p_{de,s}(\mathbf{\hat s}|\mathbf{\hat z}, \phi)\right \rbrack 
+\lambda\cdot \mathrm{CE}\left \lbrack p(\mathbf{x}|\mathbf{\hat{z}})||p_{de,o}(\mathbf{\hat x}|\mathbf{\hat z}, \psi)\right \rbrack\rbrace +K 
\label{lemma}\\
&\mathcal{L}(\theta, \phi, \psi;\mathbf{X}^N,\mathbf{S}^N)=
    \frac{1}{N}\sum_{j=1}^{N} \left\lbrack \sum_{m=1}^{M} -p(\mathbf{s}_j=m|\mathbf{\hat{z}}_j)\log p_{de,s}(\mathbf{\hat s}_j=m|\mathbf{\hat z}_j, \phi)
    + \lambda \cdot  ||\mathbf{x}_j- f_\psi(\mathbf{\hat z}_j)||_2^2 \right\rbrack 
    \label{empirical obj}
\end{align}
\hrulefill
\vspace*{-0.6cm}
\end{figure*}

\section{Problem Formulation}

Fig. \ref{semantic} shows the semantic communication system we focus on in this paper.
There is a source data $\mathbf{X}$, associated with its unknown semantic information $\mathbf{S}$, to be communicated through a noisy channel for both data recovery and semantic processing. There is also a knowledge base which stores a sufficient number of empirical data tuples of $(\mathbf{x},\mathbf{s})$, and is available to both the transmitter and the receiver. 
Unlike the classic approach of separately designing coding and modulation, the proposed joint coding-modulation is designed as a neural network followed by a random coding. The NN is to map source data $\mathbf{X}$ to the transfer probability while the random coding generates discretely modulated symbol $\mathbf{Z}$ based on this transfer probability. Accordingly, the receiver consists of two neural networks to map from channel output $\hat{\mathbf{Z}}$ to the reconstructed source data $\hat{\mathbf{X}}$ and reconstructed semantic information $\hat{\mathbf{S}}$, respectively. 

We consider a mandatory requirement for digital communication in our system, where BPSK modulation is used for simplicity. High-order modulation shall be considered in the future work. That is to say, the channel input, denoted as $\mathbf{Z}=(Z_1, Z_2, ..., Z_n)$, is an $n$-length sequence taking values from $\left\lbrace-1, +1 \right\rbrace ^{n}$. Hence the channel output can be written as $\mathbf{\hat Z}=\mathbf{Z}+\boldsymbol{\varepsilon}$, where $\boldsymbol{\varepsilon}\sim \mathcal{N}(\mathbf{0}, \sigma ^{2}\mathbf{I}_n)$ is the channel noise. 

Therefore, a Markov chain can be established as $\mathbf{S}\to \mathbf{X}\to \mathbf{Z}\to \mathbf{\hat Z}\to (\mathbf{\hat S},\mathbf{\hat X})$. As shown in Fig. \ref{semantic}, the transition between the observable source data $\mathbf{X}$ and the channel input $\mathbf{Z}$ includes coding and modulation, which are jointly realized by a neural network. This transition is modeled as a stochastic joint encoder-modulator $p_{en}(\mathbf{z}|\mathbf{x}, \theta)$ in our paper, where $\theta$ is the parameter of the neural network. As a remark, a functional encoder or modulator can be regarded as a special case of the stochastic ones. At the receiver end, we use two stochastic demodulator-decoders, namely $p_{de,o}(\mathbf{\hat x}|\mathbf{\hat z}, \psi)$ and $p_{de,s}(\mathbf{\hat s}|\mathbf{\hat z}, \phi)$ for the reconstruction of the observable source data and the semantic information respectively, where the parameters of these two networks are denoted as $\psi$ and $\phi$.

Our goal is to find a JCM network which can preserve as much information about the semantic information and the source data as possible in the channel output $\mathbf{\hat Z}$, and the corresponding receiver networks which can best recover the semantic information and source data. Therefore, a mutual information based theoretic objective function (shorted as ``MI-OBJ'') can be established as follows:
\begin{equation}
\setcounter{equation}{1}
    I(\mathbf{S};\mathbf{\hat Z}) + \lambda I(\mathbf{\mathbf{X};\hat Z}),
    \label{obj}
\end{equation}
where $\lambda$ is a trade-off hyperparameter. The MI-OBJ function is also used in many previous works of semantic communications \cite{zhang,nesct}. By taking the channel output $\mathbf{\hat{Z}}$ into consideration, the joint encoder-modulator can have a better fitting with the channel states. However, this function is sometimes not easy to directly optimize in practical scenarios, and does not take the receiver into consideration. Hence a solution is given in our Section III, which gives an alternative approximation of MI-OBJ with the receiver design also merged in.

\section{Variational Learning for Joint Coding-Modulation}
As mentioned, in some circumstances optimizing (1) is not ``operational''. As pointed out by Alemi \emph{et al}. in \cite{intractable}, it is difficult to estimate the joint distribution $p(\mathbf{x}, \mathbf{\hat z})$ and $p(\mathbf{s}, \mathbf{\hat z})$ when $\mathbf{X}$, $\mathbf{Z}$ and $\mathbf{\hat Z}$ are high dimensional variables, making MI-OBJ in (1) intractable. Therefore, we hereby use the technique of variational inference\cite{variational} to establish its lower bound and maximize this lower bound instead, which is stated as our main theorem in Section III-A. Based on this lower bound, we further derive a loss function for image transmission, as a corollary in Section III-B. 
Moreover, we apply Gumble-max method for generating random coding based modulation to fit with neural network training, which is discussed in Section III-C. 

\subsection{Variational Inference Lower Bound}
\begin{theorem}[VILB] \label{th1}
A variational inference lower bound of MI-OBJ in (1) is given by \eqref{elbo}
, where $K=H(\mathbf{S})+\lambda\cdot H(\mathbf{X})$ is a constant, $\lambda$ is a trade-off hyperparameter and $\mathrm{CE}\left \lbrack \cdot \right \rbrack$ is the Cross Entropy between two probability distributions.
\end{theorem}

Theorem \ref{th1} gives an operational lower bound of (1), where the neural networks, namely the parameters of $\theta$, $\phi$ and $\psi$, can be trained over the optimization process of this lower bound. According to \cite{intractable}, the cross entropy on conditional distributions is much easier to be sampled and estimated than the joint distribution and the mutual information function. 

To prove Theorem \ref{th1}, first 
we need to prove the following lemma, which is an intermediate result that merges the receiver design into MI-OBJ. 

\begin{lemma} A simpler lower bound of MI-OBJ is given by \eqref{lemma},
where $K=H(\mathbf{S})+\lambda\cdot H(\mathbf{X})$.
\label{le1}
\end{lemma}
\begin{proof}
First, considering the term $I(\mathbf{S};\mathbf{\hat Z})$, we have:
\begin{align}
\setcounter{equation}{4}
    I(\mathbf{S};\mathbf{\hat Z})
    &=-H(\mathbf{S}|\mathbf{\hat Z}) + H(\mathbf{S}) \nonumber\\
    &=\mathbb{E}_{p(\mathbf{s},\mathbf{\hat z})} \log \left \lbrack p(\mathbf{s}|\mathbf{\hat z}) \frac{p_{de,s}(\mathbf{\hat s}|\mathbf{\hat z}, \phi)}{p_{de,s}(\mathbf{\hat s}|\mathbf{\hat z}, \phi)}\right \rbrack + H(\mathbf{S}) \nonumber\\
    &\overset{(a)}{=}\mathrm{KL}\left \lbrack p(\mathbf{s}|\mathbf{\hat z})||p_{de,s}(\mathbf{\hat s}|\mathbf{\hat z}, \phi)\right \rbrack \nonumber \\
    &\ \ \ \ +\mathbb{E}_{p(\mathbf{s},\mathbf{\hat z})} \log p_{de,s}(\mathbf{\hat s}|\mathbf{\hat z}, \phi) + H(\mathbf{S}) \nonumber \\
    &\overset{(b)}{\ge} \mathbb{E}_{p(\mathbf{s},\mathbf{\hat z})} \log p_{de,s}(\mathbf{\hat s}|\mathbf{\hat z}, \phi) + H(\mathbf{S}) \label{same} \\
    &\overset{(c)}{=} -\mathbb{E}_{p(\mathbf{\hat z})} \mathrm{CE}\left \lbrack p(\mathbf{s}|\mathbf{\hat{z}})|| p_{de,s}(\mathbf{\hat s}|\mathbf{\hat z}, \phi)\right \rbrack + H(\mathbf{S}), \label{sz}
\end{align}
where $(a)$ follows the definition of KL divergence, and the semantic decoder $p_{de,s}(\mathbf{\hat s}|\mathbf{\hat z}, \phi)$ is expanded as a variational approximation to the true posterior $p(\mathbf{s}|\mathbf{\hat z})$; $(b)$ follows the non-negative property of KL divergence; $(c)$ follows the definition of Cross Entropy.

Similarly, we have:
\begin{equation}
    I(\mathbf{X};\mathbf{\hat Z}) \ge -\mathbb{E}_{p(\mathbf{\hat z})} \mathrm{CE}\left \lbrack p(\mathbf{x}|\mathbf{\hat{z}})|| p_{de,o}(\mathbf{\hat x}|\mathbf{\hat z}, \psi)\right \rbrack + H(\mathbf{X}).
    \label{xz}
\end{equation}

Putting \eqref{sz} and \eqref{xz} together finishes the proof of Lemma \ref{le1}.
\end{proof}

To further prove Theorem 1, we consider the probability of $\mathbf{\hat z}$, which is determined by the stochastic encoder:
\begin{align}
    p(\mathbf{\hat z})=&
    \int _{\mathbf{s},\mathbf{x},\mathbf{z}} p(\mathbf{s},\mathbf{x},\mathbf{z}, \mathbf{\hat z}) d\mathbf{s}d\mathbf{x}d\mathbf{z} \nonumber \\
    =&\int _{\mathbf{s},\mathbf{x},\mathbf{z}} p(\mathbf{s},\mathbf{x})p_{en}(\mathbf{z}|\mathbf{x}, \theta)p(\mathbf{\hat z|\mathbf{z}})d\mathbf{s}d\mathbf{x}d\mathbf{z}. \label{pz}
\end{align}

For channel input symbols $\mathbf{z}=(z_1, z_2, ..., z_n)$, we model them as mutually conditionally independent Bernoulli variables like that in \cite{nesct} with $z_i \sim B(p^{i}(\mathbf{x}, \theta))$, where $p^{i}(\mathbf{x}, \theta)$ is the probability of $z_i=1$. Due to the assumed independency among the components in $\mathbf{z}$, their joint probability mass function (PMF) can be express as:
\begin{equation}
    p_{en}(\mathbf{z}|\mathbf{x}, \theta)=\prod_{i=1}^{n} \lbrack p^{i}(\mathbf{x}, \theta)\rbrack ^{\frac{1}{2}(1+z_i)}\lbrack (1-p^{i}(\mathbf{x}, \theta))\rbrack ^{\frac{1}{2}(1-z_i)}. \label{en}
\end{equation}

By introducing \eqref{pz} and the transition probability \eqref{en} into Lemma \ref{le1}, we can finish the proof of Theorem \ref{th1}.

\subsection{Loss Function Design for Image Semantic Communications}

Based on Theorem \ref{th1}, a corresponding loss function for image semantic communication can be further derived.  

\begin{corollary}[Loss function for image source data]
\label{co1}
Consider image semantic communication scenarios with image classification as semantics, the loss function can be written as \eqref{empirical obj}
, where $(\mathbf{X}^N,\mathbf{S}^N)=\left \lbrace(\mathbf{x}_j, \mathbf{s}_j)\right\rbrace_{j=1}^{N}$ denotes the batch of training data; $M$ is the total number of classes; $f_\psi(\mathbf{\hat z}_j)$ is the mean of the Gaussian-parameterized stochastic decoder for the observable source data reconstruction, which comes from a special assumption for image source data that we will later elaborate; $||\cdot||$ denotes the $l_2$ norm;
$\mathbf{z}_j=y_\theta(\mathbf{x}_j,\mathbf{g}_j)$
is the constellation symbol randomly generated by source data $\mathbf{x}_j$ and auxiliary random variable $\mathbf{g}_j$, which will be introduced later in Subsection III-C. 
\end{corollary}

\begin{proof}
We first introduce some widely recognized assumptions of image source data, and then we further derive the loss function by applying these assumptions to \eqref{elbo} and replacing the probability function with the empirical distribution.

As pointed out in \cite{nesct}, the source data $\mathbf{X}$ and its reconstruction $\mathbf{\hat{X}}$ can be modeled as multivariate factorized Gaussian distributions with isotropic covariance. Thus, we can assume that 
\begin{align}
  &   p(\mathbf{x|\hat z})= \mathcal{N}(\boldsymbol{\mu},\sigma_1^2\mathbf{I}_{m\times m}) \label{true posterior},\\
&p_{de,o}(\mathbf{\hat x|\hat z},\psi)= \mathcal{N}(f_{\psi}(\mathbf{\hat z}),\sigma_2^2\mathbf{I}_{m\times m}), \label{variational}
\end{align}
where $\boldsymbol{\mu}$ is the true pixel value of the source data (usually normalized between $0$ and $1$), $\sigma_1$ and $\sigma_2$ are precision parameters treated as constants, and $f_{\psi}(\mathbf{\hat z})$ is the output of the stochastic decoder for the reconstruction of the source data. 

By taking \eqref{true posterior} and \eqref{variational} into Theorem \ref{th1}, \eqref{mse} can be further obtained. It shows the variational lower bound of $I(\mathbf{X};\mathbf{\hat Z})$, where $(a)$ follows the same step as in (\ref{same}), and $p_{de,o}(\mathbf{\hat x}|\mathbf{\hat z}, \psi)$ is the neural network parameterized variational approximation to the true posterior $p(\mathbf{x}|\mathbf{\hat z})$; $(b)$ follows the expansion of the expectation; $(c)$ follows the repeated use of the equation $\mathbf{x}^{\mathrm{T}}\mathrm{A}\mathbf{x}=\mathrm{tr}(\mathrm{A}\mathbf{x}\mathbf{x}^{\mathrm{T}})$ with $a$ and $b$ being two constants.

The semantic information is set as the class of the images, and thus is categorically distributed. 
By replacing probability function with empirical distribution, we have
\begin{align}
    \frac{1}{N}\sum_{j=1}^{N}   \sum_{m=1}^{M} -p(\mathbf{s}_j=m|\mathbf{\hat{z}}_j)\log p_{de,s}(\mathbf{\hat s}_j=m|\mathbf{\hat z}_j, \phi)& \\
   \rightarrow \mathbb{E}_{p(\mathbf{\hat z})} \mathrm{CE}\left \lbrack p(\mathbf{s}|\mathbf{\hat{z}})|| p_{de,s}(\mathbf{\hat s}|\mathbf{\hat z}, \phi)\right \rbrack& \nonumber\\
     \frac{1}{N}\sum_{j=1}^{N}  || \mathbf{x}_j- f_\psi(\mathbf{\hat z}_j)||_2^2 \rightarrow \mathbb{E}_{p(\mathbf{\hat z_j})}||\boldsymbol{\mu}_1-f_\psi(\mathbf{\hat{z}})||_2^2 &  ,
\end{align}
where ``$\rightarrow$'' denotes the convergence due to the law of large numbers, and $M$ denotes the total number of classes. Thus, we complete the proof of Corollary \ref{co1}.
\end{proof}

As a remark, in Corollay \ref{co1}, the loss function is a weighted sum of two terms. One term is the cross entropy between semantic information and its reconstruction, and the other is the mean square error of the reconstructed observable data. Thus, it is a loss function based on data distortion function, which has been widely used in previous works such as \cite{gunduz1, zhang}. Essentially, Corollary \ref{co1} shows that the data distortion function based loss function in \eqref{empirical obj} is a lower bound of the 
MI-OBJ in (1), when specified to image semantic communication scenarios. This conclusion gives a more concrete bridge between these two types of loss functions. Meanwhile, since deriving Corollary \ref{co1} uses the property of image data, it also partially explains why in text semantic communications the loss functions in \eqref{empirical obj} is not often used. 
\begin{figure*}[!t]
\normalsize
\setcounter{equation}{14}
\begin{align}
\label{mse}
    I(\mathbf{X};\mathbf{\hat Z})&=-H(\mathbf{X}|\mathbf{\hat Z})+H(\mathbf{X}) \nonumber\\ 
    &\overset{(a)}{\ge} \mathbb{E}_{p(\mathbf{\hat z})}\mathbb{E}_{p(\mathbf{x}|\mathbf{\hat z})} \log p_{de,o}(\mathbf{\hat x}|\mathbf{\hat z}, \psi) + H(\mathbf{X}) \nonumber\\
    &\overset{(b)}{=} \int_{\mathbf{\hat{z}}}p(\mathbf{\hat z}) d \mathbf{\hat z}\int_{\mathbf{x}}\frac{1}{(2\pi )^{\frac{m}{2}} \sigma_{1}}e^{-\frac{1}{2}(\mathbf{x}-\mathbf{\boldsymbol{\mu}})^{T}\frac{1}{\sigma_1^{2}}\mathbf{I}(\mathbf{x}-\mathbf{\boldsymbol{\mu}})}\log \frac{1}{(2\pi )^{\frac{m}{2}} \sigma_{2}}e^{-\frac{1}{2}(\mathbf{\hat x}-f_\psi(\mathbf{\hat{z}}))^{T}\frac{1}{\sigma_2^{2}}\mathbf{I}(\mathbf{\hat x}-f_\psi(\mathbf{\hat{z}}))}d \mathbf{x} + H(\mathbf{X})\nonumber \\
    &\overset{(c)}{=} \int_{\mathbf{\hat{z}}}p(\mathbf{\hat z})  (b-a\left\lbrack\boldsymbol{\mu}_1-f_\psi(\mathbf{\hat{z}}))^T(\boldsymbol{\mu}_1-f_\psi(\mathbf{\hat{z}}))\right\rbrack d \mathbf{\hat z} + H(\mathbf{X}) \nonumber\\
    &\overset{(d)}{=}-a \cdot \mathbb{E}_{p(\mathbf{\hat z})}
    ||\boldsymbol{\mu}_1-f_\psi(\mathbf{\hat{z}})||_2^2
    +b + H(\mathbf{X})
\end{align}
\hrulefill
\vspace*{-0.6cm}
\end{figure*}
\subsection{Gumble-max Method for Random Coding}

In order to fit with the neural network training process, a Gumble-max method is applied, along with its corresponding gradient estimation method. Hence in this subsection, the random code generation and gradient estimation by using Gumble-max method will be discussed. 

$\bullet$ \emph{Random code generation}

To get a codeword $\mathbf{z}$, we need to sample from the stochastic encoder $p_{en}(\mathbf{z}|\mathbf{x},\theta)$, so that the empirical distribution of the random code matches its probability distribution. It has been proven true for using Gumbel-max method \cite{gumbelproof}. 
Specifically, in Gumbel-max we denote the sampling process as $\mathbf{z}= y_\theta(\mathbf{x}, \mathbf{g})\nonumber$,
where each symbol $z_i$ is generated by the following rule:
\begin{align}
\setcounter{equation}{13}
    z_i=\underset{\left\lbrace+1, -1\right\rbrace}{\mathrm{argmax}}\left\lbrace\log p^i(\mathbf{x},\theta) + g^{i}_1, \log \left \lbrack 1-p^i(\mathbf{x},\theta)\right \rbrack+ g^{i}_2\right\rbrace,
    \label{mapping}
\end{align}
where $g^{i}_1$ and $g^{i}_2$ are independent Gumbel variables with distribution
\begin{equation}
g^{i}_1,g^{i}_2\sim \mathrm{Gumbel}(0,1).     \nonumber
\end{equation}

In operation, we first draw a sample from $u\sim \mathrm{Uniform}(0,1)$, then use the inverse transform and generate the Gumbel-distributed variable as $g=-\log(-\log(u))$\cite{gumbel}.

$\bullet$ \emph{Gradient estimation}

Estimating the gradient of the lower bound (\ref{elbo}) w.r.t the encoder parameters $\theta$ is tricky, since for one thing, the sampling process is non-differentiable; for another, the usual Monte Carlo gradient estimator exhibits high variance\cite{vae}. 

To solve the first problem, the non-differentiable argmax operation in \eqref{mapping} will be approximated by the Softmax funtion\cite{gumbel}. The reparameterization trick widely applied in VAE\cite{vae} solves the second problem. Through the additional Gumbel noise $\mathbf{g}$ and the transformation $\mathbf{z}= y_\theta(\mathbf{x}, \mathbf{g})$, the Monte Carlo gradient estimator of $\theta$ w.r.t the expectation of an arbitrary function $h(\mathbf{z})$ now becomes:
\begin{align}
    \nabla_\theta\mathbb{E}_{p_{en}(\mathbf{z}|\mathbf{x},\theta)}\left\lbrack h(\mathbf{z})\right\rbrack&=\nabla_\theta\mathbb{E}_{p(\mathbf{g})}\left\lbrack h(y_\theta(\mathbf{x}, \mathbf{g}))\right\rbrack \nonumber\\
    &=\mathbb{E}_{p(\mathbf{g})}\nabla_\theta\left\lbrack h(y_\theta(\mathbf{x}, \mathbf{g}))\right\rbrack \nonumber\\
    &\simeq\frac{1}{N}\sum_{m=1}^{N}\nabla_\theta \left\lbrack h(y_\theta(\mathbf{x}, \mathbf{g}_m))\right\rbrack. \nonumber
\end{align}
The expectation is no longer with respect to the encoder parameters $\theta$ but the Gumbel noise, resulting a low-variance gradient estimator\cite{vae}.

\section{Experiment Results}

In this section, we use simulations to test the performance of our proposed method. The experiments are conducted on CIFAR10 dataset, where the class of the images is set as the semantic information. Three conventional methods are used as the benchmarks, including two digital modulation methods and one analog modulation method.
\subsection{Experiment Settings}
\subsubsection{Datasets}
Our experiments are performed on the CIFAR10 dataset which includes 10 classes of $32\times 32$ color images with a training set of 50000 examples and a test set of 10000 examples. 

\subsubsection{Benchmarks}
We compare the performance of our proposed joint coding-modulation scheme with BPSK modulation with three benchmarks: the neural network based analog modulation (abbreviated as ``Analog''), the 8-bit quantizer for BPSK modulation (abbreviated as ``8-bit Uniform'') and learning-based quantization for BPSK modulation (abbreviated as ``1-bit NN''). The details of the baseline methods are listed as follows.
\begin{itemize}
\item \emph{``Analog'' method}: This method uses a neural network with four Resnet blocks\cite{resnet} for encoding, a 4-layer Spinal\cite{spinal} architecture for semantic information decoding and a 3-layer trans-convolution architecture for source data decoding. The output of the encoding network is a sequence of real numbers, and then is sent into AWGN channel directly. As a remark, this method does not obey the digital communication requirement of our system model. 
\item \emph{``8-bit Uniform'' method}: A same neural network architecture as Analog method is used here. However, the output of the encoding network is sent to an 8-bit uniform quantizer to convert into bit sequence. Then the ``0'' bits are associated with ``-1'', and the ``1'' bits are associated with ``+1'' in BPSK modulation.
\item \emph{``1-bit NN'' method}: A same neural network architecture as Analog method is used for joint source-channel coding. Then an one-layer MLP with a Sigmoid nonlinearity is applied as the 1-bit quantizer, which outputs the probability of $b_i=1$ for each bit. A hard decision is then applied to the output (if the probability is over 0.5, then this bit is decided as 1) so as to attain a bit sequence. An association from bit sequence to BPSK constellations are also made, which is similar to the 8-bit method.
\end{itemize}

\subsubsection{Neural Network of Our Method}
The neural network architecture of our method is the same as the analog method. All details of loss function design, training methods and gradient estimation follows Section III. 

\subsection{Performance Comparison}
\begin{figure}[htbp]
\centering
\vspace{-0.6cm}
\setlength{\abovecaptionskip}{-0.2cm}
\includegraphics[width=0.76\linewidth]{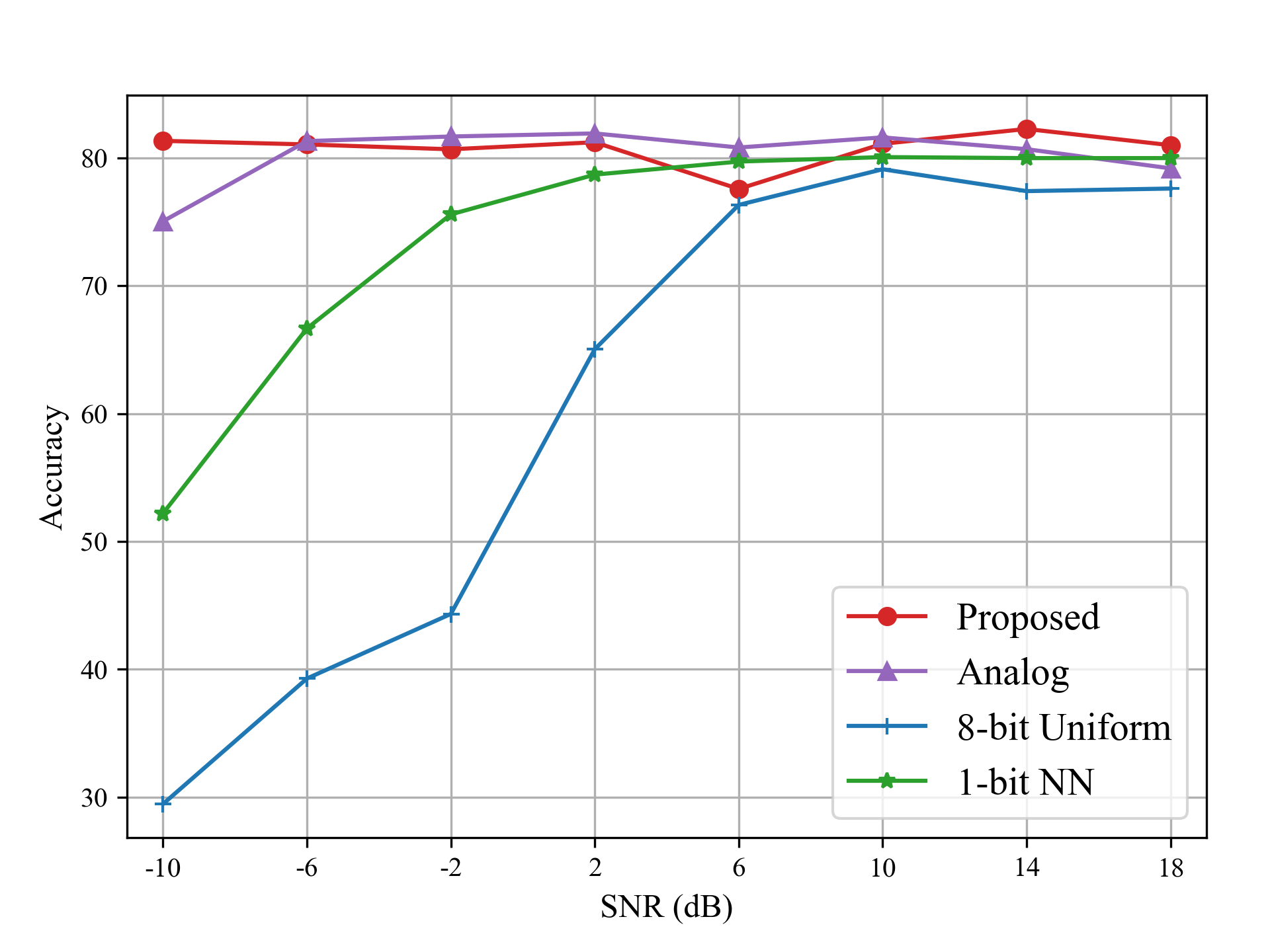}
\caption{Accuracy against SNR (code length n=1536).}
\label{snr acc}
\vspace{-0.6cm}
\end{figure}

\begin{figure}[htbp]
\centering
\vspace{-0.2cm}
\setlength{\abovecaptionskip}{-0.2cm}
\includegraphics[width=0.76\linewidth]{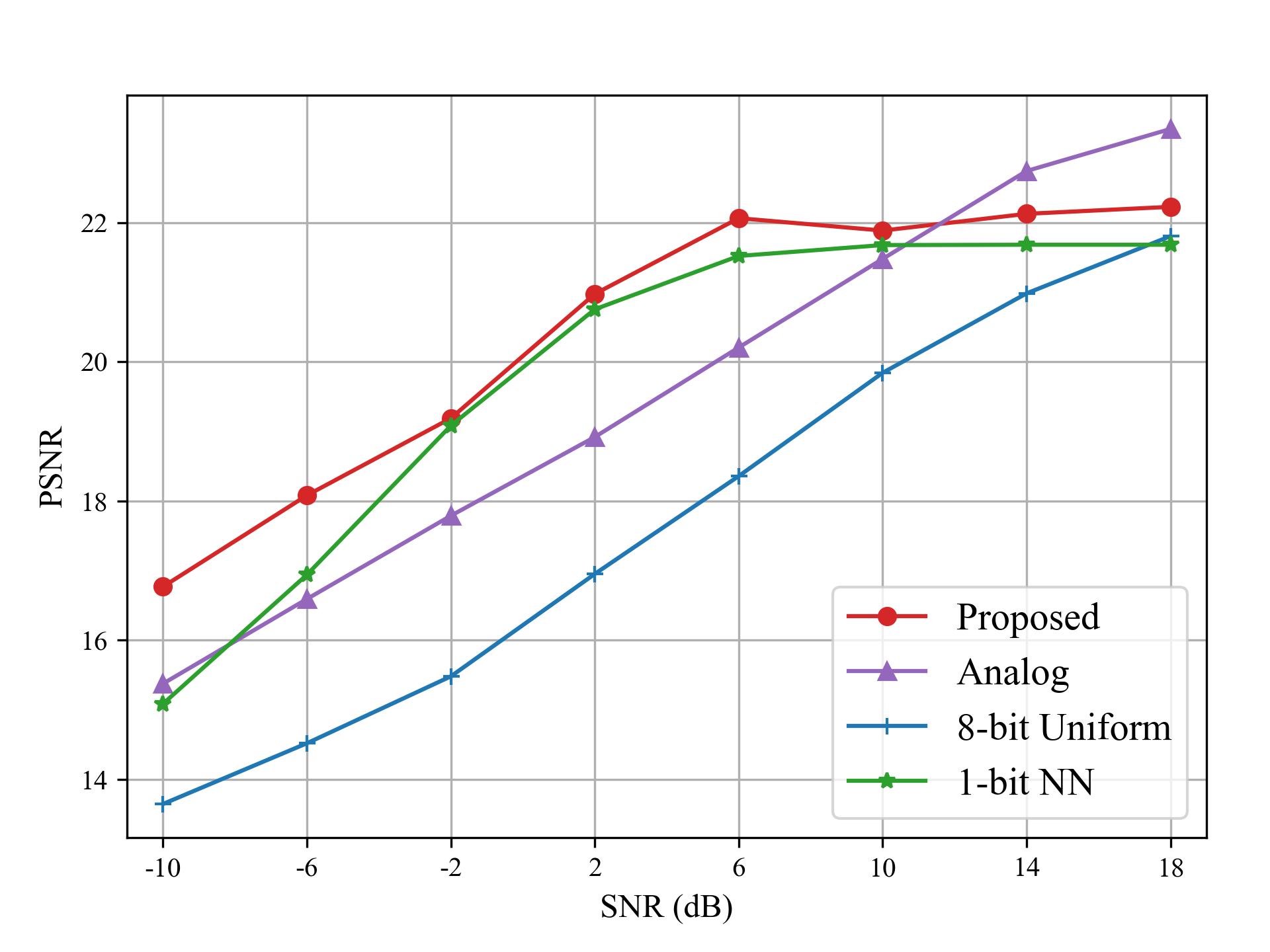}
\caption{PSNR against SNR (code length n=1536).}
\label{snr psnr}
\vspace{-0.15cm}
\end{figure}

Fig. \ref{snr acc} and Fig. \ref{snr psnr} plot the performance of all four methods against SNR, with channel input sequence length $n$=1536. As a general trend, our method outperforms baseline methods when SNR is low. Meanwhile, when SNR is high, the performance of all methods will converge, but our method still remains competitive, only slightly less than the Analog method.

Fig. \ref{snr acc} shows the semantic information accuracy performance against SNR. It can be seen that our method has stable performance with SNR changing, while the performance of all other methods will be largely affected if SNR is low. Therefore our method has an obvious performance advantage in low SNR regime. For example, when SNR=-10 dB, the accuracy of our method can reach $81.32\%$ while the second place Analog method can only achieves $73.63\%$.

Fig. \ref{snr psnr} shows the PSNR of observable source data reconstruction against SNR. Similar with the accuracy performance, our method keeps a performance advantage in most cases. Only when SNR=14 dB and SNR=18 dB, the PSNR performance of our method is slightly worse than the Analog method. When SNR=14 dB, the PSNR of our method is $22.13$ dB while $22.74$ dB for Analog method; when SNR=18 dB, the PSNR of our method is $22.23$ dB while $23.34$ dB for Analog method. 

Fig. \ref{n acc} and Fig. \ref{n psnr} plot the performance of all four methods against channel input sequence length $n$ where the channel SNR is set to be -2 dB. Generally, for both semantic information reconstruction and observable source data reconstruction, our method always outperforms three baseline methods. 

Fig. \ref{n acc} shows the semantic information accuracy performance against sequence length. It can be seen that when SNR=-2 dB, the accuracy performance of our method is close to that of Analog method, and greatly exceeds the other two methods. When $n$=512, our proposed method reaches a accuracy performance of $74.67\%$ while the 1-bit NN method and 8-bit Uniform method respectively have a performance of $68.49\%$ and $34.49\%$. 

Fig. \ref{n psnr} shows the PSNR performance of reconstructing observable source data against sequence length, for all four methods. As we can see, our method outperforms all three benchmarks in all code length. When SNR=-2 dB and $n$=512, the gap between our proposed method and the second place 1-bit NN method is about $0.8$ dB.

\begin{figure}[htbp]
\centering
\vspace{-0.5cm}
\setlength{\abovecaptionskip}{-0.2cm}
\includegraphics[width=0.76\linewidth]{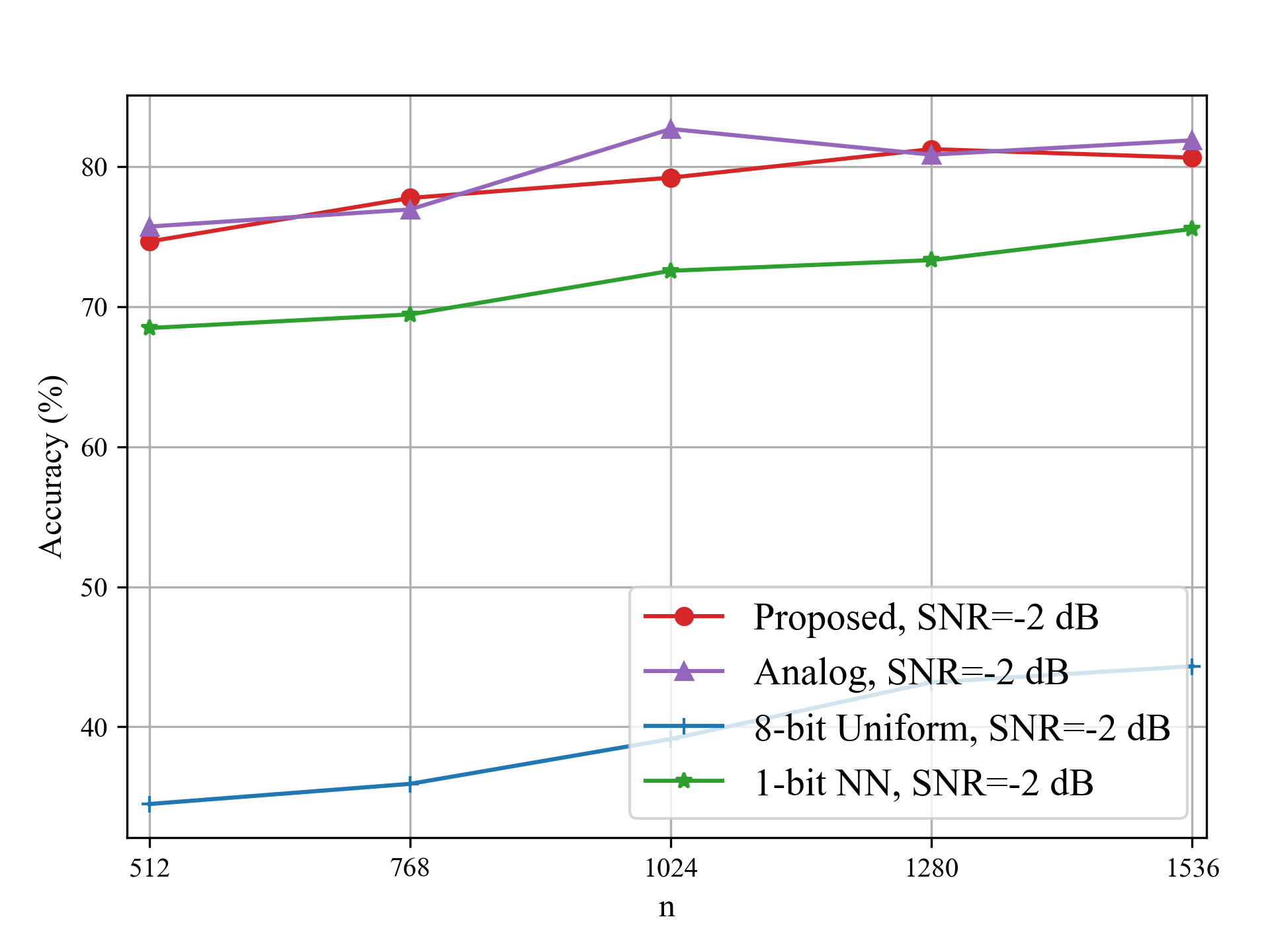}
\caption{Accuracy against code length n.}
\label{n acc}
\vspace{-0.4cm}
\end{figure}
\begin{figure}[htbp]
\centering
\vspace{-0.5cm}
\setlength{\abovecaptionskip}{-0.2cm}
\includegraphics[width=0.76\linewidth]{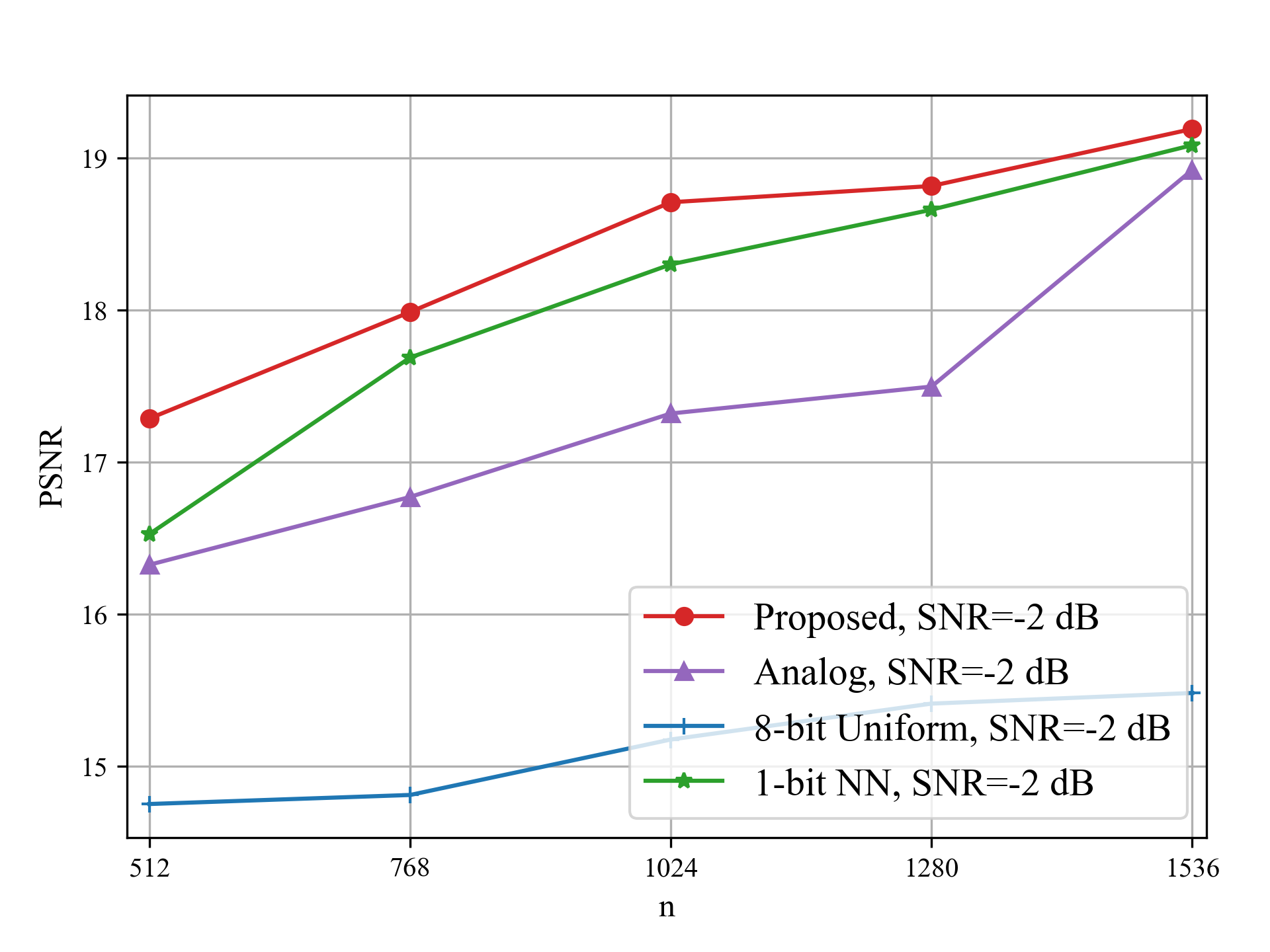}
\caption{PSNR against code length n.}
\label{n psnr}
\vspace{-0.6cm}
\end{figure}

\section{Conclusion}
In this paper, we propose a learning-based joint coding-modulation scheme for semantic communications with digital modulation. Experiment results show that our method can outperform all existing methods for digital semantic communications. Meanwhile, new loss functions are proposed, which are more operational for high dimensional data and reveal the relationship between mutual information based loss function and data distortion function based loss function. 

The method we proposed can be further extended to higher-order digital modulation systems, such as QPSK or 8QAM, which will be explored in our future work. 



\bibliographystyle{IEEEtran}
\bibliography{ref}{}

\end{document}